\documentclass{IEEEtran}
\usepackage{graphicx,color}
\usepackage{cite}
\usepackage{setspace} 
\usepackage{amsmath}
\usepackage{amsmath,amsthm,amssymb}
\usepackage{multirow}
\usepackage{hhline}
\usepackage{epsfig}
\usepackage{subfigure}
\usepackage{epstopdf}
\usepackage{verbatim}
\usepackage{algorithm}
\usepackage{algorithmicx}
\usepackage{algpseudocode}
\usepackage{etoolbox}
\usepackage{cases}
\usepackage{csquotes}
\usepackage{mathtools,stmaryrd}
\SetSymbolFont{stmry}{bold}{U}{stmry}{m}{n}
\usepackage{bbm}
\usepackage{lettrine}

\newcommand{\suchthat}{\;\ifnum\currentgrouptype=16 \middle\fi|\;}

\newtheorem{theorem}{Theorem}

\newtheorem{corr}{Corollary}

\usepackage[papersize={210mm,297mm}, left=0.6in, right=0.6in, top=0.75in, bottom=1in]{geometry}

\allowdisplaybreaks

\begin{document}

\title{On the Level Crossing Rate\\of  Fluid Antenna Systems}

\author{\IEEEauthorblockN{Priyadarshi Mukherjee, Constantinos Psomas, and~Ioannis Krikidis}

\IEEEauthorblockA{Department of Electrical and Computer Engineering, University of Cyprus\\
Email: \{mukherjee.priyadarshi, psomas, krikidis\}@ucy.ac.cy}}

\maketitle
\begin{abstract}
Multiple-input multiple-output (MIMO) technology has significantly impacted wireless communication, by providing extraordinary performance gains. However, a minimum inter-antenna space constraint in MIMO systems does not allow its integration in devices with limited space. In this context, the concept of fluid antenna systems (FASs) appears to be a potent solution, where there is no such restriction. In this paper, we investigate the average level crossing rate (LCR) of such FASs. Specifically, we derive closed-form analytical expressions of the LCR of such systems and extensive Monte-Carlo simulations validate the proposed analytical framework. Moreover, we also demonstrate that under certain conditions, the LCR obtained coincides with that of a conventional selection combining-based receiver. Finally, the numerical results also provide insights regarding the selection of appropriate parameters that enhance the system performance.
\end{abstract}

\begin{IEEEkeywords}
Fluid antenna systems, spatial correlation, level crossing rate, selection combining, diversity.
\end{IEEEkeywords}

\section{Introduction}
Multiple-input multiple-output (MIMO) can be considered as one of the most popular wireless technologies in recent years. The concepts of diversity and multiplexing gain form the basis of MIMO systems, which led to its extraordinary performance for wireless communication links. However, there must be a minimum distance of $\frac{\lambda}{2}$ between the antennas in MIMO systems, where $\lambda$ is the transmission wavelength \cite{goldsmith2005wireless}. This limits the integration of MIMO systems inside mobile devices such as tablets and mobile phones, where the physical space is very limited.

To overcome the aforementioned limitation of MIMO, the novel concept of fluid antenna system (FAS) was proposed in \cite{fluid1}. FAS is essentially a single antenna system with $N$ fixed locations (referred as ``port'') distributed over a given space. The idea of FAS is originally motivated by the increasing trend of using ionized solutions or liquid metals for antennas \cite{lant1,lant2,lant3}. The most interesting aspect of FAS is that an antenna element is no longer kept fixed at a particular location, but it can switch to a relatively more favorable location inside the boundaries, if required. Apparently, the objective of FAS resembles that of traditional transmit antenna selection (TAS) systems \cite{aselec}, where multiple antennas are deployed at different locations and the antenna with the strongest signal is selected. However, unlike TAS systems, the single antenna element in FAS can change position among the predetermined ports. In this way, an FAS exploits the phenomenon of spatial diversity and the received signal from the port with the strongest channel condition is selected. Furthermore, there is no limitation of maintaining a minimum inter-port distance and as a result, the space making up the FAS may be small with large number of ports. The work in \cite{fluid1} evaluates the theoretical performance of such systems in terms of outage probability. It is shown that even with a small space and a practically feasible number of ports, FAS can significantly outperform conventional maximum ratio combining-based systems. The work in \cite{fluid2} investigates the second order statistics of a FAS-based receiver. In particular, the work evaluates its performance in terms of the level crossing rate (LCR), average fade duration (AFD), and also ergodic capacity.

As stated above, FAS does not have any constraint on the inter-port distance. This, unlike in conventional MIMO systems, makes the aspect of spatial correlation a crucial factor in characterizing its performance limits. However, the work in \cite{fluid2} does not make this consideration. Motivated by this, in this work, we investigate the FAS second order statistics by taking the spatial correlation into account. Specifically, we derive closed-form analytical expressions of the LCR as a function of both the number of ports and the associated spatial correlation. We demonstrate that in certain scenarios, the FAS LCR coincides with that of a conventional selection combining (SC)-based system with independent and identically distributed (i.i.d.) channels. To the best of our knowledge, this is the first work that presents a complete analytical framework to characterize the LCR for a FAS, by taking into account its practical constraints and limitations.

\section{System model}
\begin{figure}[!t]
 \centering\includegraphics[width=0.95\linewidth]{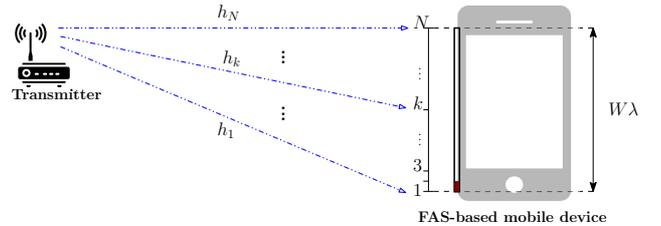}
\caption{The considered topology consisting of a single antenna transmitter and a FAS-based mobile device.}
\label{fig:model}
\end{figure}

We consider a simple point-to-point topology, where we have a single antenna transmitter and a $N$-port FAS-based mobile receiver, as depicted in Fig. \ref{fig:model}. As shown in the figure, a typical FAS-based receiver is essentially a single antenna system with a single radio-frequency (RF) chain. The $N$ ports are evenly distributed in a linear space of $W\lambda$, where $\lambda$ is the transmission wavelength and the antenna can switch locations instantly among the ports \cite{fluid1}. Fig. \ref{fig:model} illustrates that the first port is the reference and the liquid antenna is always switched to the port with best channel conditions. The channels at these ports are characterized as \cite{fluid1}
\begin{align}
\begin{cases}
\!h_{1}=\sigma x_{0} +j\sigma y_{0}\\
\!h_{k}=\sigma \left ({\sqrt {1-\mu _{k}^{2}}x_{k}+\mu _{k} x_{0}}\right) \\
\quad\:\: +j\sigma \left ({\sqrt {1-\mu _{k}^{2}}y_{k}+\mu _{k} y_{0}}\right) \:\: \text {for }k=2, {\cdots },N, \end{cases}
\end{align}
where $x_0,\cdots,x_N,y_0,\cdots,y_N$ are independent Gaussian random variables with zero mean and variance $\frac{1}{2}$, and $\mu_k$ $\forall$ $k$ are the parameters that control the spatial correlation between the channels. Accordingly, we have $\mathbb{E}[|h_k|^2]=\sigma^2$ $\forall$ $k$ and \cite{fluid1}
\begin{equation}  \label{bess}
\mu_k=J_0 \left( \frac{2\pi (k-1)}{N-1}W \right), \qquad \text {for }k=2, {\dots },N,
\end{equation}
where $J_0(\cdot)$ is the zero-order Bessel function of the first kind. Note that due to the oscillatory nature of $J_0(\cdot)$, for a given $N$, $\mu_k$ $\forall$ $k$ is not a monotonically decreasing function of $W$. It may happen that, multiple values of $W$ result in identical $\mu_k$. Hence, to obtain a reasonable range of $W,$ we consider $J_0(\cdot)$ till the point where it reaches the first zero. Accordingly, irrespective of $N$, we obtain $W \in [0,0.38]$. The FAS always selects the port with the strongest channel condition, i.e.
\begin{equation}
|h_{\mathrm{ FAS}}|=\max \{|h_{1}|,|h_{2}|, {\dots },|h_{N}|\}.
\end{equation}
As the ports in an FAS are very close to each other, the aspect of spatial correlation plays a crucial role in determining which of the $N$ ports is selected. Accordingly, the joint probability distribution function (PDF) of $|h_1|,|h_2|,\cdots,|h_N|$ is \cite{fluid1}
\begin{align} \label{jpdf1}
&p_{|h_{1}|,\cdots,|h_{N}|}(x_{1},\dots,x_{N}) \nonumber \\
&\quad =\prod _{\substack{k=1\\ (\mu _{1}\triangleq 0)}}^{N}\!\!\frac {2x_{k}}{\sigma ^{2}(1-\mu _{k}^{2})}e^{-\frac {x_{k}^{2}+\mu _{k}^{2}x_{1}^{2}}{\sigma ^{2}(1-\mu _{k}^{2})}}I_{0}\left ({\frac {2~\mu _{k}x_{1}x_{k}}{\sigma ^{2}(1-\mu _{k}^{2})}}\right)\!,
\end{align}
for $x_{1}, {\dots },x_{N}\ge 0$, where $I_0(\cdot)$ is the zero-order modified Bessel function of the first kind. It is important to note that the mutual coupling does not affect an FAS, as only one antenna element is activated at each time. Hence, \eqref{jpdf1} is not a conventional $N$-variate random variable, but a product of $N$ bi-variate random variables.

\section{Level Crossing Rate for FAS}  \label{afda}

In this section, we characterize the LCR of an FAS, which is an important parameter in characterizing the dynamics of any random process. It facilitates to evaluate the impact of the time-varying channel on the FAS performance. The LCR enables to estimate the statistics of error occurrence in signal detection. The LCR of a random process $r$ at threshold $r_{\rm th}$ essentially gives the number of times per unit duration that $r$ crosses $r_{\rm th}$ in the negative (or positive) direction \cite{rice}. Mathematically it defined as
\begin{equation}    \label{lcrdef}
L (r_{\rm th})=\int_0^{\infty}\dot{r}p_{\dot{R}R}(\dot{r},r_{\rm th})d\dot{r},
\end{equation}
where $\dot{r}$ is the time derivative of $r$ and $p_{\dot{R}R}(\dot{r},r)$ is the joint PDF of $r(t)$ and $\dot{r}(t)$ in an arbitrary instant $t$. For an isotropic scattering scenario, the time derivative of the signal envelope is Gaussian distributed with zero mean, irrespective of the fading distribution \cite{stuber}. As we aim to analyze the LCR of a FAS, we propose the following theorem in this direction. 
\begin{theorem}  \label{theo1}
The LCR for a $N$-port FAS is given by \eqref{lcrp}.
\end{theorem}
\begin{figure*} [t]
\begin{align}  \label{lcrp}
L(x_{\rm th})&=\frac{\sqrt{2\pi}x_{\rm th}f_D}{\sigma} \Biggl\{ e^{-\frac{x_{\rm th}^2}{\sigma ^2}} \prod _{k=2}^N \left[ 1-Q_1 \left( \sqrt{\frac{2\mu_k^2}{\sigma ^{2}(1-\mu _{k}^{2})}} x_{\rm th} , \sqrt{\frac{2}{\sigma ^{2}(1-\mu _{k}^{2})}} x_{\rm th} \right) \right] + \sum_{i=2}^N \frac {1}{(1-\mu _i^{2})}e^{-\frac {x_{\rm th}^{2}}{\sigma ^{2}(1-\mu _i^{2})}} \nonumber \\
& \quad \times \int_0^{x_{\rm th}} \frac{2x_1}{\sigma^2}e^{-\frac{x_1^2}{\sigma ^{2}(1-\mu _i^{2})}} I_{0}\left ({\frac {2~\mu _ix_{\rm th}x_1}{\sigma ^{2}(1-\mu _i^{2})}}\right) \prod _{\substack{k=2 \\ k \neq i}}^N \left[ 1-Q_1 \left( \sqrt{\frac{2\mu_k^2}{\sigma ^{2}(1-\mu _{k}^{2})}} x_1 , \sqrt{\frac{2}{\sigma ^{2}(1-\mu _{k}^{2})}} x_{\rm th} \right) \right]dx_1  \Biggr\}.
\end{align}
\hrule
\end{figure*}
\begin{proof}
See Appendix A.
\end{proof}
We observe from \eqref{lcrp} that the LCR is a function of the spatial correlation, the number of ports, the decision threshold, and the maximum Doppler frequency of the channel. This  LCR of an FAS is different from that of a conventional SC-based receiver, primarily because of the unique PDF of the channels at the $N$ ports (as it can be seen from \eqref{jpdf1}) and also the aspect of  the associated spatial correlation. For the sake of completeness, we consider the following two extreme cases: $\mu_k=0,1$ $\forall$ $k=1,\cdots,N$. The LCR $L(x_{\rm th})$ corresponding to $\mu_k=0$ $\forall$ $k$, i.e. for a no spatial correlation scenario, is given below.
\begin{corr}  \label{cor1}
For a scenario without spatial correlation, i.e. $\mu_k=0$ $\forall$ $k$, $L(x_{\rm th})$ is
\begin{equation}  \label{coin}
L(x_{\rm th}) = N\sqrt{2\pi} f_D \frac{x_{\rm th}}{\sigma}e^{-\frac{x_{\rm th}^2}{\sigma^2}} \left( 1-e^{-\frac{x_{\rm th}^2}{\sigma^2}} \right)^{N-1}.
\end{equation}
\end{corr}
\noindent The above corollary follows directly from Theorem \ref{theo1}, by replacing $\mu_k=0$ $\forall$ $k=1,\cdots,N$ and using
\begin{equation}
Q_1(0,b)=\int_b^{\infty}xe^{-\frac{x^2}{2}}dx=e^{-\frac{b^2}{2}} \quad \text{for} \quad b \geq 0.
\end{equation}
Without a spatial correlation, for a given set of system parameters, $L(x_{\rm th})$ obtained in \eqref{coin} coincides with the LCR of a conventional SC-based receiver with i.i.d. channels \cite[Eq. 18]{slctn}.
The case with $\mu_k\!=\!1$ $\forall $ $k$ essentially corresponds to the scenario where the $N$ ports are identical, i.e., there is no need of any switching of the liquid antenna among the ports.

\begin{corr}  \label{mu1}
For a scenario with $\mu_k\!=\!1$ $\forall $ $k$, we have
\begin{equation}
L(x_{\rm th})=\frac{\sqrt{2\pi}}{\sigma}f_Dx_{\rm th}e^{-\frac{x_{\rm th}^2}{\sigma^2}}.
\end{equation}
\end{corr}
\begin{proof}
See Appendix B.
\end{proof}
\noindent It is interesting to note from the above corollary, that in case of identical channels, the FAS LCR is independent of $N$. 
Nevertheless, the analytical expression of $L(x_{\rm th})$ derived in Theorem \ref{theo1}, being too involved, does not provide any insightful analysis. Hence, we consider a simple case of $N=2$.

\begin{corr}  \label{corr2}
For a two-port FAS, the LCR is given by
\begin{align}  \label{n2f}
L(x_{\rm th})&=\frac{2\sqrt{2\pi} f_Dx_{\rm th}}{\sigma ^{3}(1-\mu^{2})} e^{\left(\!\!-\frac {x_{\rm th}^2}{\sigma ^{2}(1-\mu^2)} \right)} \nonumber \\
&\!\!\!\!\!\!\!\!\!\!\!\!\!\!\!\!\!\! \times \sum_{k=0}^{\infty} \frac{\left(\mu x_{\rm th}\right)^{2k}}{(k!)^2 \left( \sigma ^{2}(1-\mu^2) \right)^{k-1}}\gamma \left( k+1, \frac{x_{\rm th}^2}{\sigma ^{2}(1-\mu^2)} \right).
\end{align}
\end{corr}

\begin{proof}
See Appendix C.
\end{proof}

\noindent The above corollary demonstrates the effect of parameters such as $x_{\rm th}, \mu,$ and $\sigma$ on the LCR. For example, we observe from \eqref{n2f} that $L(x_{\rm th})$ is the product of an unimodal function and an increasing function with respect to $x_{\rm th}$. This implies that $L(x_{\rm th})$ is also unimodal, i.e. $L(x_{\rm th})$ initially increases with $x_{\rm th}$, but it starts decreasing after a certain point.

\section{Numerical Results}
We now validate our theoretical analysis with extensive Monte-Carlo simulations. Without any loss of generality, we consider unit power channels, i.e. $\mathbb{E}[|h_k|^2]=\sigma^2=1$ $\forall$ $k=1,\cdots,N$, where $N$ is the number of ports in the FAS and a carrier frequency of $900$ MHz.

\begin{figure}[!t]
 \centering\includegraphics[width=0.95\linewidth]{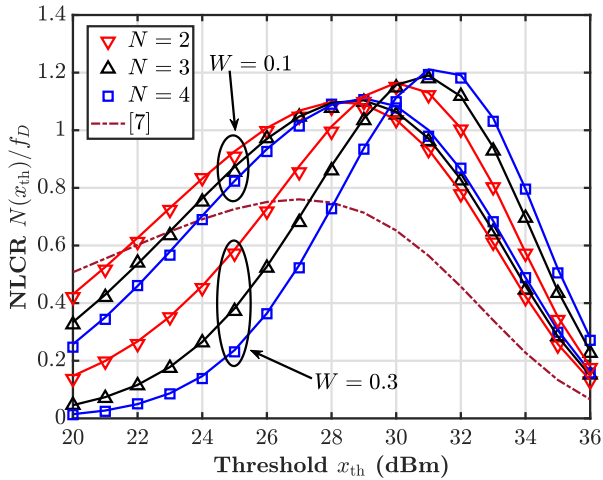}
\caption{Verification of proposed analysis via Monte-Carlo simulations.}
\label{fig:fig1}
\end{figure}

Fig. \ref{fig:fig1} demonstrates the variation of the normalized LCR (NLCR) $L(x_{\rm th})/f_D$ versus the decision threshold $x_{\rm th}$ for two scenarios with $W=0.1$ and $0.3$, respectively. In this figure, we have considered a $N$-port FAS, with $N=2,3,$ and $4$, respectively. Note that the values considered are solely for illustration. We observe that the theoretical results (lines) match very closely with the simulation results (markers); this verifies our proposed analytical framework. The figure supports our claim that LCR depends on both $x_{\rm th}$ and $\mu$. It can be further noted that the LCR increases with increase of $x_{\rm th}$ until it reaches its maximum and then it decreases. Moreover, this particular value of $x_{\rm th}$ depends on the value of $W$, which corroborates the claims made in \cite{blet} regarding the effect of optimum threshold selection. Furthermore, we observe that, the choice of $W$ also has an impact on the NLCR. Finally, the closed-form expression for NLCR derived in \cite{fluid2} significantly deviates from the simulation results.

\begin{figure}[!t]
\centering\includegraphics[width=0.95\linewidth]{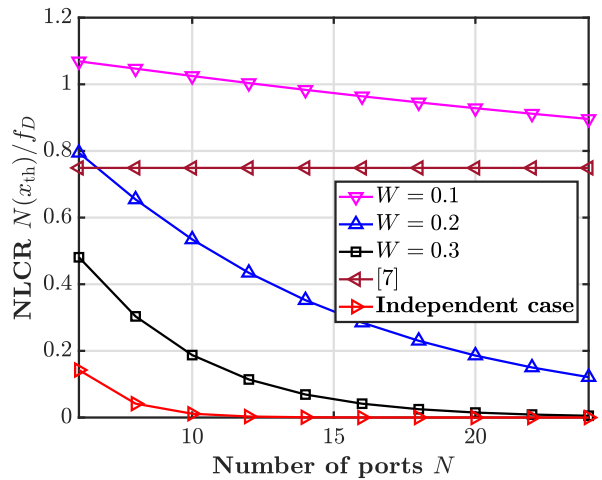}
\caption{Impact of spatial correlation on NLCR; $x_{\rm th}=28$ dBm.}
\label{fig:fig2}
\end{figure}

Fig. \ref{fig:fig2} illustrates the effect of $N$ on the NLCR performance. For an identical $N$, a marginal increase of $W$ leads to a significant improvement in performance; for example, observe the performance gap at $N=16$ between $W=0.1,0.2,$ and $0.3,$ respectively. It is worth to note, that the best performance is observed when the channels are independent at the ports, i.e. $\mu_k=0$ $\forall$ $k=1,\cdots,N$. Finally, the figure demonstrates that irrespectively of the choice of $W$, an FAS asymptotically attains the optimal performance in terms of NLCR, as $N \rightarrow \infty$; greater the value of $W$, faster is the convergence. However, from \eqref{bess}, we know that $W$ cannot be increased arbitrarily due to the oscillatory nature of $J_0(\cdot)$; in this way, the spatial correlation affects the FAS performance. Furthermore, we also observe that the change of $W$ or $N$ does not affect the NLCR, as in \cite{fluid2}.

Fig. \ref{fig:fig3} depicts the variation of NLCR with $N$ for multiple values of $x_{\rm th}$. We observe that for both the cases, i.e. \eqref{lcrp} and \cite{fluid2}, a lower threshold results in a lower NLCR. Furthermore, we observe that, as also seen in Fig. \ref{fig:fig2}, \cite{fluid2} is invariant to $N$ (from the dashed lines). On the contrary, NLCR as derived in \eqref{lcrp} is significantly affected with increasing $N$; the NLCR decreases with $N$. This demonstrates the key benefit of an FAS, where it is advantageous to have higher $N$ without any inter-port distance constraint.

\section{Conclusion}
In this paper, motivated by the practical constraints of a FAS, we proposed a novel and general analytical framework for the exact evaluation of an important second order statistical parameter of FAS, namely the LCR. In particular, by considering the effect of the time-varying nature of fading channels, we investigated the aspect of spatial correlation in characterizing this performance metric. Closed form expressions for the LCR were analytically derived and it was demonstrated that in certain scenarios, they coincide with the LCR of an SC-based receiver with i.i.d. channels. Finally, we validate our proposed framework by extensive Monte-Carlo simulations. Based on the framework presented in this paper, an immediate extension of this work is to investigate the second order statistics of a multiple FAS-based topology with a heterogeneous geometry.


\section*{Appendix}

\subsection{Proof of Theorem \ref{theo1}}  \label{app1}

\begin{figure}[!t]
 \centering\includegraphics[width=0.95\linewidth]{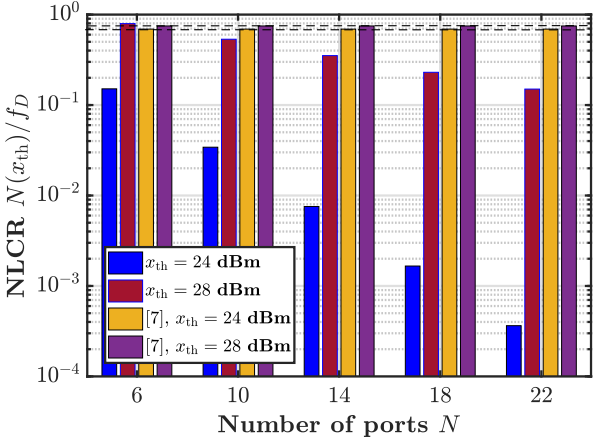}
\caption{Impact of spatial correlation on NLCR; $W=0.2$.}
\label{fig:fig3}
\end{figure}

In this context, the $N$-variate joint PDF $p_{|\dot{h}|,|h|}(\dot{x},x)$ is given by \cite[8.42]{jpdf}
\begin{align}  \label{alou}
p_{|\dot{h}|,|h|}(\dot{x},x)&=\sum_{i=1}^N p_{|\dot{h_i}|}(\dot{x}) \nonumber \\[-0.2em]
& \!\! \times \underbrace{\int_0^x\!\!\cdots\!\!\int_0^x}_{(N-1)-{\rm fold}} \!\!\!\! p_{|h_{1}|,\cdots,|h_{N}|}(x_1,\dots,x_i=x,\cdots,x_N) \nonumber \\[-0.2em]
& \times \underbrace{dx_1 \cdots dx_k\cdots dx_N}_{\substack{(N-1)-{\rm fold} \\ k \neq i}},
\end{align}
where $|\dot{h}_i|$ is the time derivative of the signal envelope at the $i$-th port. Hence, from \eqref{lcrdef}, we obtain the LCR as
\begin{align}  \label{step1}
L(x_{\rm th})&=\int_0^{\infty}\dot{x}p_{|\dot{h}|,|h|}(\dot{x},x)d\dot{x} \nonumber\\
&=\int_0^{\infty}\!\!\!\!\dot{x} \sum_{i=1}^N p_{|\dot{h_i}|}(\dot{x})\!\! \nonumber \\[-0.2em]
& \times \underbrace{\int_0^{x_{\rm th}}\!\!\!\!\cdots\!\!\int_0^{x_{\rm th}}}_{(N-1)-{\rm fold}}\!\!\!\! \!\!p_{|h_{1}|,\cdots,|h_{N}|}(x_1,\dots,x_i=x_{\rm th},\cdots,x_N) \nonumber \\[-0.2em]
& \times \underbrace{dx_1 \cdots dx_k\cdots dx_N}_{\substack{(N-1)-{\rm fold} \\ k \neq i}}d\dot{x}.
\end{align}
Furthermore, in case of an identically distributed Rayleigh fading scenario, $p_{|\dot{h_i}|}(\dot{x})$ $\forall$ $i=1,\cdots,N,$ follows a zero mean Gaussian PDF \cite{stuber} with variance $\sigma_{\dot{X}}^2=\pi^2\sigma^2f_D^2$, where $f_D$ is the maximum Doppler frequency. As a result, we obtain
\begin{equation}  \label{pdot}
\int_0^{\infty}\!\!\dot{x}p_{|\dot{h_i}|}(\dot{x})d\dot{x}=\frac{\sigma_{\dot{X}}}{\sqrt{2\pi}}=\sqrt{\frac{\pi}{2}}\sigma f_D, \:\: \forall i=1,\cdots,N.
\end{equation}
By combining \eqref{step1} and \eqref{pdot}, we have
\begin{align}  \label{step2}
&L(x_{\rm th})=\sqrt{\frac{\pi}{2}}\sigma f_D \nonumber \\[-0.2em]
& \times  \sum_{i=1}^N \underbrace{\int_0^{x_{\rm th}}\!\!\cdots\int_0^{x_{\rm th}}}_{(N-1)-{\rm fold}}p_{|h_{1}|,\cdots,|h_{N}|}(x_1,\dots,x_i=x_{\rm th},\cdots,x_N) \nonumber \\[-0.2em]
& \times  \underbrace{dx_1 \cdots dx_k\cdots dx_N}_{\substack{(N-1)-{\rm fold} \\ k \neq i}}.
\end{align}
The summation term in \eqref{step2} is alternatively written as
\begin{align}  \label{altr}
&\sum_{i=1}^N \underbrace{\int_0^{x_{\rm th}}\!\!\cdots\int_0^{x_{\rm th}}}_{(N-1)-{\rm fold}}p_{|h_{1}|,\cdots,|h_{N}|}(x_1,\dots,x_i=x_{\rm th},\cdots,x_N) \nonumber \\[-0.2em]
& \quad \times \underbrace{dx_1 \cdots dx_k\cdots dx_N}_{\substack{(N-1)-{\rm fold} \\ k \neq i}} \nonumber \\[-0.2em]
&= \underbrace{\int_0^{x_{\rm th}}\!\!\cdots\int_0^{x_{\rm th}}}_{(N-1)-{\rm fold}}p_{|h_{1}|,\cdots,|h_{N}|}(x_1=x_{\rm th},\cdots,x_N) \underbrace{dx_2 \cdots dx_N}_{(N-1)-{\rm fold}} \nonumber \\[-0.2em]
& \quad + \sum_{i=2}^N \underbrace{\int_0^{x_{\rm th}}\!\!\cdots\!\!\int_0^{x_{\rm th}}}_{(N-1)-{\rm fold}}\!\!p_{|h_{1}|,\cdots,|h_{N}|}(x_1,\dots,x_i=x_{\rm th},\cdots,x_N) \nonumber \\[-0.2em]
& \quad \times \underbrace{dx_1 \cdots dx_k\cdots dx_N}_{\substack{(N-1)-{\rm fold} \\ k \neq i}}.
\end{align}
The first term of \eqref{altr} is evaluated as
\begin{align}  \label{1st}
&\underbrace{\int_0^{x_{\rm th}}\cdots\int_0^{x_{\rm th}}}_{(N-1)-{\rm fold}}p_{|h_{1}|,\cdots,|h_{N}|}(x_1=x_{\rm th},\cdots,x_N) \underbrace{dx_2 \cdots dx_N}_{(N-1)-{\rm fold}} \nonumber \\[-0.2em]
&\overset{(a)}{=}\frac{2x_{\rm th}}{\sigma ^2}e^{-\frac{x_{\rm th}^2}{\sigma ^2}}\int_0^{x_{\rm th}}\cdots\int_0^{x_{\rm th}}\prod _{k=2}^N
\frac {2x_{k}}{\sigma ^{2}(1-\mu _{k}^{2})}e^{-\frac {x_{k}^{2}+\mu _{k}^{2}x_{\rm th}^{2}}{\sigma ^{2}(1-\mu _{k}^{2})}} \nonumber \\[-0.2em]
& \quad \times I_{0}\left ({\frac {2~\mu _{k}x_{\rm th}x_{k}}{\sigma ^{2}(1-\mu _{k}^{2})}}\right) dx_2 \cdots dx_N \nonumber \\[-0.2em]
& =\frac{2x_{\rm th}}{\sigma ^2}e^{-\frac{x_{\rm th}^2}{\sigma ^2}} \prod _{k=2}^N \int_0^{x_{\rm th}} \frac {2x_{k}}{\sigma ^{2}(1-\mu _{k}^{2})}e^{-\frac {x_{k}^{2}+\mu _{k}^{2}x_{\rm th}^{2}}{\sigma ^{2}(1-\mu _{k}^{2})}} \nonumber \\[-0.2em]
& \quad \times I_{0}\left ({\frac {2~\mu _{k}x_{\rm th}x_{k}}{\sigma ^{2}(1-\mu _{k}^{2})}}\right) dx_k \nonumber \\
&\overset{(b)}{=}\frac{2x_{\rm th}}{\sigma ^2}e^{-\frac{x_{\rm th}^2}{\sigma ^2}} \prod _{k=2}^N \left[ 1-Q_1 \left( \sqrt{\frac{2\mu_k^2}{\sigma ^{2}(1-\mu _{k}^{2})}} x_{\rm th} , \right.\right. \nonumber \\[-0.2em]
& \quad \left.\left. \sqrt{\frac{2}{\sigma ^{2}(1-\mu _{k}^{2})}} x_{\rm th} \right) \right],
\end{align}
where $(a)$ follows from \eqref{jpdf1}, $Q_1(\cdot,\cdot)$ is the first-order Marcum $Q$-function, and $(b)$ follows from \cite[Eq. 10]{marcum}. Hereafter, the second term of \eqref{altr} is expanded as \eqref{2nd}, where $(c)$ is based on \cite[Eq. 10]{marcum} and the fact that the joint PDF stated in \eqref{jpdf1} is not a regular multivariate Rayleigh PDF. This distribution is a product of $N$ pairs of bi-variate Rayleigh PDFs, with the first port of the FAS being the reference point for all the remaining $N-1$ ports. By combining \eqref{step2}, \eqref{altr}, \eqref{1st}, and \eqref{2nd}, we obtain \eqref{lcrp}, which completes the proof.
\begin{figure*}[t]
\begin{align}  \label{2nd}
& \sum_{i=2}^N \underbrace{\int_0^{x_{\rm th}}\!\!\cdots\int_0^{x_{\rm th}}}_{(N-1)-{\rm fold}}p_{|h_{1}|,\cdots,|h_{N}|}(x_1,\dots,x_i=x_{\rm th},\cdots,x_N) \underbrace{dx_1 \cdots dx_k\cdots dx_N}_{\substack{(N-1)-{\rm fold} \\ k \neq i}} \nonumber \\[-0.2em]
&= \sum_{i=2}^N \int_0^{x_{\rm th}}\!\!\cdots\!\!\int_0^{x_{\rm th}} \!\!\frac {2x_{\rm th}}{\sigma ^{2}(1-\mu _i^{2})}e^{-\frac {x_{\rm th}^{2}+\mu _i^{2}x_1^{2}}{\sigma ^{2}(1-\mu _i^{2})}}I_{0}\left ({\frac {2~\mu _ix_{\rm th}x_1}{\sigma ^{2}(1-\mu _i^{2})}}\right)\!\! \prod _{\substack{k=1 \\ k \neq i}}^N  \frac {2x_{k}}{\sigma ^{2}(1-\mu _{k}^{2})}e^{-\frac {x_{k}^{2}+\mu _{k}^{2}x_1^{2}}{\sigma ^{2}(1-\mu _{k}^{2})}}I_{0}\left ({\frac {2~\mu _{k}x_kx_1}{\sigma ^{2}(1-\mu _{k}^{2})}}\right) dx_1 \cdots dx_k\cdots dx_N  \nonumber \\[-0.2em]
&\overset{(c)}{=} \sum_{i=2}^N \frac {2x_{\rm th}}{\sigma ^{2}(1-\mu _i^{2})}e^{-\frac {x_{\rm th}^{2}}{\sigma ^{2}(1-\mu _i^{2})}}\!\! \int_0^{x_{\rm th}} \frac{2x_1}{\sigma^2}e^{-\frac{x_1^2}{\sigma ^{2}(1-\mu _i^{2})}} I_{0}\left (\!{\frac {2~\mu _ix_{\rm th}x_1}{\sigma ^{2}(1-\mu _i^{2})}}\!\right)\!\! \prod _{\substack{k=2 \\ k \neq i}}^N \left[\! 1-Q_1 \left( \sqrt{\frac{2\mu_k^2}{\sigma ^{2}(1-\mu _{k}^{2})}} x_1 , \sqrt{\frac{2}{\sigma ^{2}(1-\mu _{k}^{2})}} x_{\rm th} \right)\! \right]dx_1.
\end{align}
\vspace{-1mm}
\hrule
\end{figure*}
\vspace{-1mm}
\subsection{Proof of Corollary \ref{mu1}}
\label{appmu1}
The case of $\mu_k=1$ $\forall$ $k$ implies identical channels at all the ports. As a result, we obtain
\begin{align}
L(x_{\rm th})&=\int_0^{\infty} \!\! \dot{x}p_{|\dot{h}|,|h|}(\dot{x},x_{\rm th})d\dot{x}\overset{(a)}{=}p_{|h|}(x_{\rm th})\!\! \int_0^{\infty} \dot{x}p_{|\dot{h}|}(\dot{x})d\dot{x} \nonumber \\[-0.2em]
&\!\!\!\!\!\!\!\!\!\!\!\!=\frac{2x_{\rm th}}{\sigma^2}e^{-\frac{x_{\rm th}^2}{\sigma^2}}\!\!\int_0^{\infty}\!\!\dot{x}p_{|\dot{h_i}|}(\dot{x})d\dot{x}\overset{(b)}{=}\frac{\sqrt{2\pi}}{\sigma}f_Dx_{\rm th}e^{-\frac{x_{\rm th}^2}{\sigma^2}},
\end{align}
\noindent where $(a)$ follows from $p_{|\dot{h}|,|h|}(\dot{x},x)=p_{|\dot{h}|}(\dot{x})p_{|h|}(x)$ \cite[Eq. 2.97]{stuber} and $(b)$ follows from \eqref{pdot}. Hence, the proof. 


\subsection{Proof of Corollary \ref{corr2}}  \label{app2}

Since we are considering a two-port scenario, we take into account the joint PDF of the channel at these two ports. Hence, by replacing $N=2$ in \eqref{jpdf1}, the joint PDF becomes
\begin{align}  \label{pdf2var}
p_{|h_1|,|h_2|}(x_1,x_2)&=\frac {4x_1x_2}{\sigma ^{4}(1-\mu_2^{2})}e^{-\frac {x_1^2+x_2^2}{\sigma ^{2}(1-\mu_2^2)}} \nonumber \\[-0.2em]
& \!\!\!\!\!\!\times I_{0}\left ({\frac {2\mu_2 x_1x_2}{\sigma ^{2}(1-\mu_2^2)}}\right)\!,\quad \text{for} \quad x_1,x_2 \geq 0.
\end{align}
Thus, by replacing $N=2$ in \eqref{lcrp} and after some trivial algebraic manipulations, we obtain
\begin{align}  \label{n21}
L(x_{\rm th})&=\sqrt{\frac{\pi}{2}}\sigma f_D \left( \int_0^{x_{\rm th}} p_{|h_1|,|h_2|}(x_{\rm th},x_2)dx_2 \right. \nonumber \\[-0.2em]
& \quad \left. +\int_0^{x_{\rm th}} p_{|h_1|,|h_2|}(x_1,x_{\rm th})dx_1 \right) \nonumber \\[-0.2em]
&\overset{(a)}{=}\sqrt{2\pi}\sigma f_D \!\! \int_0^{x_{\rm th}} p_{|h_1|,|h_2|}(x_{\rm th},x_2)dx_2 \nonumber \\[-0.2em]
& = \frac{4\sqrt{2\pi}\sigma f_Dx_{\rm th}}{\sigma ^{4}(1-\mu^{2})} \exp \left(\!\!-\frac {x_{\rm th}^2}{\sigma ^{2}(1-\mu^2)} \right) \int_0^{x_{\rm th}}\!\! x_2 \nonumber \\[-0.2em]
& \quad \times \! \exp \left(\!\!-\frac {x_2^2}{\sigma ^{2}(1-\mu^2)}\!\! \right) \!\!I_{0}\left ( \!{\frac {2\mu x_{\rm th}x_2}{\sigma ^{2}(1-\mu^2)}}\!\right)\! dx_2,
\end{align}
where $(a)$ follows from \eqref{pdf2var}
and
the integral 
$\int_0^{x_{\rm th}}x_2 \! \exp \left(\!\!-\frac {x_2^2}{\sigma ^{2}(1-\mu^2)}\!\! \right) I_{0}\left ( \!{\frac {2\mu x_{\rm th}x_2}{\sigma ^{2}(1-\mu^2)}}\!\right)\! dx_2$
is obtained as
\begin{align}  \label{n22}
&\int_0^{x_{\rm th}} x_2 \exp \left(-\frac {x_2^2}{\sigma ^{2}(1-\mu^2)} \right) I_{0}\left ({\frac {2\mu x_{\rm th}x_2}{\sigma ^{2}(1-\mu^2)}}\right) dx_2 \nonumber \\[-0.2em]
&\overset{(b)}{=} \!\! \int_0^{x_{\rm th}} \!\! x_2 \exp \left(\!\!-\frac {x_2^2}{\sigma ^{2}(1-\mu^2)} \right) \!\!\sum_{k=0}^{\infty} \frac{1}{(k!)^2}\left(\frac {\mu x_{\rm th}x_2}{\sigma ^{2}(1-\mu^2)} \right)^{\!\! 2k} \!\! dx_2 \nonumber \\
&\overset{(c)}{=} \sum_{k=0}^{\infty} \frac{1}{(k!)^2} \left(\frac {\mu x_{\rm th}}{\sigma ^{2}(1-\mu^2)} \right)^{2k} \nonumber \\
& \quad \times \int_0^{x_{\rm th}} \exp \left(-\frac {x_2^2}{\sigma ^{2}(1-\mu^2)} \right) x_2^{2k+1} dx_2 \nonumber \\
&=\frac{1}{2}\! \sum_{k=0}^{\infty}\! \frac{\left(\mu x_{\rm th}\right)^{2k}}{(k!)^2 \left( \sigma ^{2}(1-\mu^2) \right)^{k-1}} \gamma \left(\! k+1, \frac{x_{\rm th}^2}{\sigma ^{2}(1-\mu^2)}\! \right).
\end{align}
Here $(b)$ follows from \cite[8.445]{grad} and by assuming $0<\mu<1$, $(c)$ follows from changing the order of summation and integration, and $\gamma(\cdot,\cdot)$ denotes the lower incomplete Gamma function. As a result, by substituting \eqref{n22} in \eqref{n21}, we finally obtain \eqref{n2f}.

\section*{Acknowledgment}

This work was co-funded by the European Regional Development Fund and the Republic of Cyprus through the Research and Innovation Foundation, under the project INFRASTRUCTURES/1216/0017 (IRIDA). It has also received funding from the European Research Council (ERC) under the European Union’s Horizon 2020 research and innovation programme (Grant agreement No. 819819).

\bibliographystyle{IEEEtran}
\bibliography{refs}

\begin{thebibliography}{10}
\providecommand{\url}[1]{#1}
\csname url@samestyle\endcsname
\providecommand{\newblock}{\relax}
\providecommand{\bibinfo}[2]{#2}
\providecommand{\BIBentrySTDinterwordspacing}{\spaceskip=0pt\relax}
\providecommand{\BIBentryALTinterwordstretchfactor}{4}
\providecommand{\BIBentryALTinterwordspacing}{\spaceskip=\fontdimen2\font plus
\BIBentryALTinterwordstretchfactor\fontdimen3\font minus
  \fontdimen4\font\relax}
\providecommand{\BIBforeignlanguage}[2]{{%
\expandafter\ifx\csname l@#1\endcsname\relax
\typeout{** WARNING: IEEEtran.bst: No hyphenation pattern has been}%
\typeout{** loaded for the language `#1'. Using the pattern for}%
\typeout{** the default language instead.}%
\else
\language=\csname l@#1\endcsname
\fi
#2}}
\providecommand{\BIBdecl}{\relax}
\BIBdecl

\bibitem{goldsmith2005wireless}
A.~Goldsmith, \emph{Wireless communications}.\hskip 1em plus 0.5em minus
  0.4em\relax Cambridge University Press, 2005.

\bibitem{fluid1}
K.-K. Wong, A.~Shojaeifard, K.-F. Tong, and Y.~Zhang, ``Fluid antenna
  systems,'' \emph{IEEE Trans. Wireless Commun.}, vol.~20, no.~3, pp.
  1950--1962, Mar. 2021.

\bibitem{lant1}
A.~Dey, R.~Guldiken, and G.~Mumcu, ``Microfluidically reconfigured wideband
  frequency-tunable liquid-metal monopole antenna,'' \emph{IEEE Trans. Antennas
  Propag.}, vol.~64, no.~6, pp. 2572--2576, June 2016.

\bibitem{lant2}
C.~Borda-Fortuny, K.-F. Tong, A.~Al-Armaghany, and K.-K. Wong, ``A low-cost
  fluid switch for frequency-reconfigurable vivaldi antenna,'' \emph{IEEE
  Antennas Wireless Propag. Lett.}, vol.~16, pp. 3151--3154, 2017.

\bibitem{lant3}
A.~Singh, I.~Goode, and C.~E. Saavedra, ``A multistate frequency reconfigurable
  monopole antenna using fluidic channels,'' \emph{IEEE Antennas Wireless
  Propag. Lett.}, vol.~18, no.~5, pp. 856--860, 2019.

\bibitem{aselec}
I.~Bahceci, T.~Duman, and Y.~Altunbasak, ``Antenna selection for
  multiple-antenna transmission systems: performance analysis and code
  construction,'' \emph{IEEE Trans. Inf. Theory}, vol.~49, no.~10, pp.
  2669--2681, Oct. 2003.

\bibitem{fluid2}
K.~K. Wong, A.~Shojaeifard, K.-F. Tong, and Y.~Zhang, ``Performance limits of
  fluid antenna systems,'' \emph{IEEE Commun. Lett.}, vol.~24, no.~11, pp.
  2469--2472, Nov. 2020.

\bibitem{rice}
S.~O. Rice, ``Mathematical analysis of random noise,'' \emph{Bell Syst. Tech.
  J.}, vol.~24, pp. 46--156, 1945.

\bibitem{stuber}
G.~L. St\"{u}ber, \emph{Principles of Mobile Communications}.\hskip 1em plus
  0.5em minus 0.4em\relax Boston: Kluwer Academic Publishers, 1996.

\bibitem{slctn}
C.-D. Iskander and P.~Takis~Mathiopoulos, ``Analytical level crossing rates and
  average fade durations for diversity techniques in {Nakagami} fading
  channels,'' \emph{IEEE Trans. Commun.}, vol.~50, no.~8, pp. 1301--1309, Aug.
  2002.

\bibitem{blet}
A.~S. Panajotović, M.~. Stefanović, D.~L. Draca, and N.~M. Sekulović,
  ``Average level crossing rate of dual selection diversity in correlated
  {Rician} fading with {Rayleigh} cochannel interference,'' \emph{IEEE Commun.
  Lett.}, vol.~14, no.~7, pp. 605--607, July 2010.

\bibitem{jpdf}
L.~Yang and M.-S. Alouini, \emph{Average outage duration of wireless
  communication systems}.\hskip 1em plus 0.5em minus 0.4em\relax Kluwer
  Academic Publishers, 2004, ch.~8, pp. 209--240.

\bibitem{marcum}
G.~Corazza and G.~Ferrari, ``New bounds for the {Marcum Q-function},''
  \emph{IEEE Trans. Inf. Theory}, vol.~48, no.~11, pp. 3003--3008, Nov. 2002.

\bibitem{grad}
I.~S. Gradshteyn and I.~M. Ryzhik., \emph{Table of {I}ntegrals, {S}eries, and
  {P}roducts}.\hskip 1em plus 0.5em minus 0.4em\relax Elsevier, 2007.

\end{thebibliography}

\end{document}